\documentclass{sig-alternate-05-2015}

\newif\iffullpaper
\fullpapertrue			

\iffullpaper
	\newcommand{\fullpaper}[1]{#1}
\else 
	\newcommand{\fullpaper}[1]{}
\fi

\usepackage{graphicx}
\usepackage[utf8]{inputenc}
\usepackage{url}
\usepackage{color,xcolor,colortbl}
\usepackage{algorithm,algorithmicx,algpseudocode}
\usepackage[mathscr]{euscript}		
\usepackage{helvet}
\usepackage{epstopdf}
\usepackage{nicefrac}
\usepackage{enumitem}

\newcommand{\etal}{et~al.}
\newcommand{\eg}{e.g.}
\newcommand{\ie}{i.e.}

\newcommand{\algFont}{\fontsize{10}{13}\selectfont}
\newcommand{\E}{\mathbf{\mathrm{E}}}
\algrenewcommand\textproc{\textsf}

\setlist[description]{listparindent=\parindent,leftmargin=0em,itemsep=1em,topsep=0.3em,font={\normalfont\sffamily\sfsize}}
\setlist[itemize]{topsep=0.3em, itemsep=0em, leftmargin=1.75em}

\newtheorem{theorem}{Theorem}
\newtheorem{lemma}{Lemma}

\algdef{SE}[DOWHILE]{Do}{DoWhile}{\algorithmicdo}[1]{\algorithmicwhile\ #1}
\algnewcommand\True{\textbf{true}}
\algnewcommand\False{\textbf{false}}

\algdef{SE}[SUBALG]{Indent}{EndIndent}{}{\algorithmicend\ }%
\algtext*{Indent}
\algtext*{EndIndent}

\newcommand{\bricks}{}
\def\bricks/{\mbox{TorBricks}}

\newcommand{\sfsize}{\fontsize{0.8\baselineskip}{0.68\baselineskip}\selectfont}
\newcommand{\sans}[1]{\textsf{\sfsize \mbox{#1}}}
\newcommand{\para}[1]{\vspace{0.55em} \noindent \sans{{\mbox{#1}}}}

\title{TorBricks: Blocking-Resistant Tor Bridge Distribution}

\numberofauthors{3}
\author{
	\alignauthor
	Mahdi Zamani\\[0.3em]
	\affaddr{Yale University}\\[0.1em]
	\affaddr{New Haven, CT}\\[0.3em]
	\email{mahdi.zamani@yale.edu}
	\alignauthor
	Jared Saia\\[0.3em]
	\affaddr{University of New Mexico}\\[0.1em]
	\affaddr{Albuquerque, NM}\\[0.3em]
	\email{saia@cs.unm.edu}	
	\and
	\alignauthor 
	Jedidiah Crandall\\[0.3em]
	\affaddr{University of New Mexico}\\[0.1em]
	\affaddr{Albuquerque, NM}\\[0.3em]
	\email{crandall@cs.unm.edu}
}

\begin{document}
\sloppy
\maketitle

\begin{abstract}
{
Tor is currently the most popular network for anonymous Internet access. It critically relies on volunteer nodes called \emph{bridges} for relaying Internet traffic when a user's ISP blocks connections to Tor. 
Unfortunately, current methods for distributing bridges are vulnerable to malicious users who obtain and block bridge addresses.
In this paper, we propose \bricks/, a protocol for distributing Tor bridges to $n$ users, even when an unknown number ${t < n}$ of these users are controlled by a malicious adversary. \bricks/ distributes $O(t\log{n})$ bridges and guarantees that all honest users can connect to Tor with high probability after $O(\log{t})$ rounds of communication with the distributor. 
\newline We also extend our algorithm to perform privacy-preserving bridge distribution when run among multiple untrusted distributors. This not only prevents the distributors from learning bridge addresses and bridge assignment information, but also provides resistance against malicious attacks from a $\lfloor m/3 \rfloor$ fraction of the distributors, where $m$ is the number of distributors.}
\end{abstract}

\maketitle

\section{Introduction}
Mass surveillance and censorship increasingly threaten democracy and freedom of speech. A growing number of governments around the world control the Internet to protect their domestic political, social, financial, and security interests~\cite{Turner:2016:Surveillance,Rushe:2012:Censorship}. Countering this trend is the rise of anonymous communication systems which strive to preserve the privacy of individuals in cyberspace. Tor~\cite{dingledine:2004} is the most popular of such systems with more than 2.5 million users on average per day~\cite{Tor:Users}. Tor relays Internet traffic via more than 6,500 volunteer nodes called \emph{relays} spread across the world~\cite{Tor:Relays}. By routing data through random paths in the network, Tor can protect the private information of its users such as identity and geographical location.


Since the list of all relays is available publicly, state-sponsored organizations can enforce Internet service providers (ISPs) to block access to all of them making Tor unavailable in territories ruled by the state. 
When access to the Tor network is blocked, Tor users have the option to use \emph{bridges}, which are volunteer relays not listed in Tor's public directory~\cite{Dingledine06designof}. Bridges serve only as entry points into the rest of the Tor network, and their addresses are carefully distributed to the users, with the hope that they cannot all be learned by censors. 
As of March 2016, about 3,000 bridge nodes were running daily in the Tor network~\cite{Tor:Bridges}.

Currently, bridges are distributed to users based on different strategies such as CAPTCHA-enabled email-based distribution and IP-based distribution~\cite{Dingledine06designof}. Unfortunately, censors are using sophisticated attacks to obtain
and block bridges, rendering Tor unavailable for many users~\cite{Dingledine:Bridges:2011,Ling:2012:infocom,BridgeBlockingChina:2012}.
Also, state of the art techniques for bridge distribution either (1) cannot provably guarantee that all honest users can access Tor~\cite{WangLBH:rBridge:13,McCoy:FC:2011,Sovran:2008:PSN}; (2) can only work when the number of corrupt users is known in advance~\cite{Mahdian:2010}; (3) require fully trusted distributors~\cite{McCoy:FC:2011,Mahdian:2010,Sovran:2008:PSN}; and/or (4) cannot resist malicious attacks from the distributors~\cite{WangLBH:rBridge:13,McCoy:FC:2011,Mahdian:2010,Sovran:2008:PSN}.

We believe that it is challenging in practice to identify or even estimate the number of corrupt users, due to the sophisticated nature of Internet censorship in many countries such as China~\cite{Oni:2012:China,Ensafi2015b}. Moreover, we believe it is desirable for a system to be robust to attacks on the distributors. For example, powerful adversaries can hack into these systems, not only to break users' anonymity by obtaining bridge assignment information, but also to prevent the bridge distribution protocol from achieving its goals.

\begin{figure*}[t]
	\centering
	\includegraphics[width=0.6\linewidth]{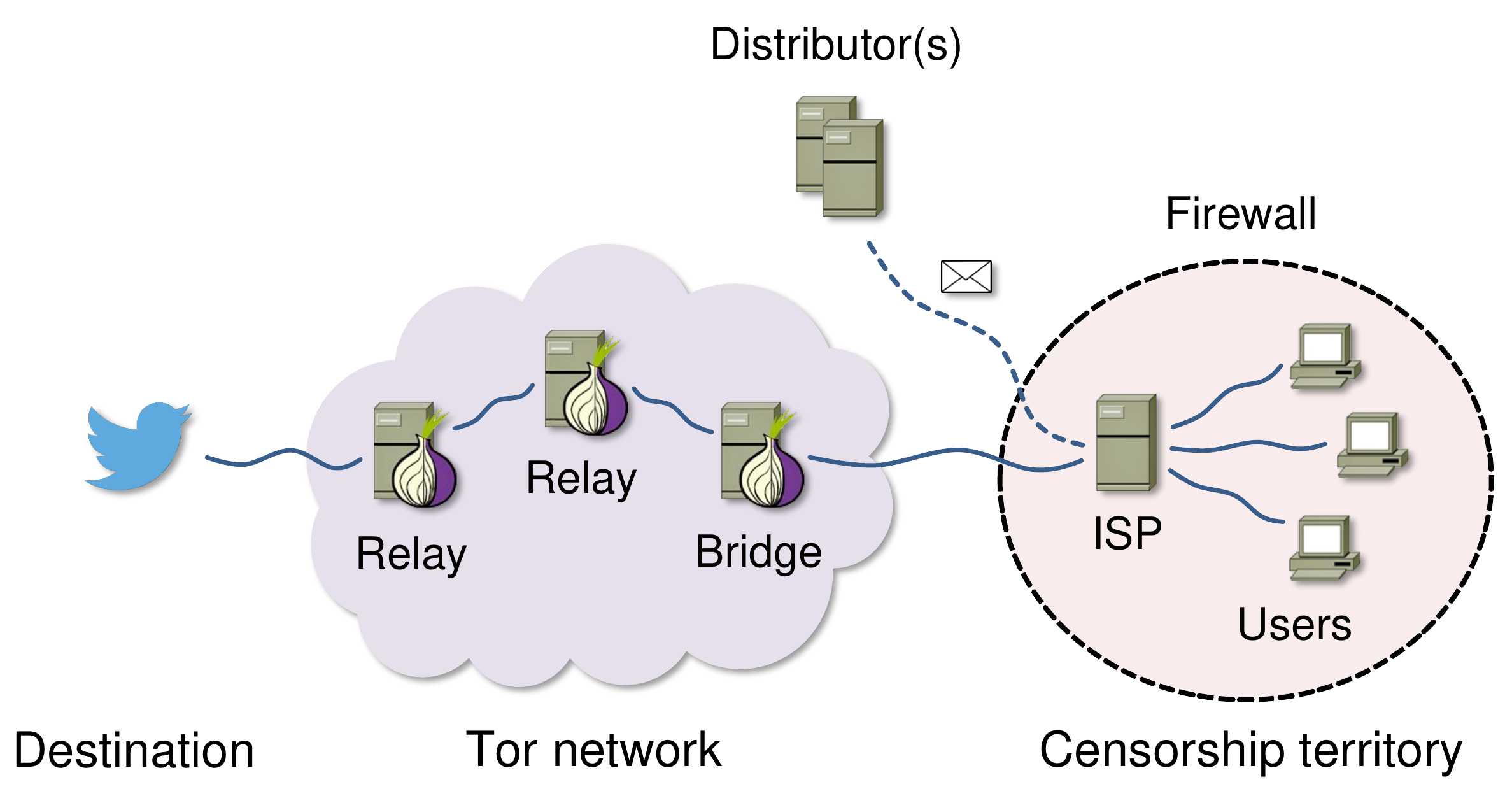}
	\caption{Our network model}
	\label{fig:model}
\end{figure*}

In this paper, we propose \bricks/, a bridge distribution algorithm that provably ensures Tor is available to all honest users with high probability, without requiring any \emph{a priori} knowledge about the number of corrupt users. \bricks/ guarantees that the number of rounds until all honest users can connect to Tor is bounded by $O(\log{t})$, where ${t<n}$ is the number of corrupt users. This is achieved by distributing at most $O(t\log{n})$ bridge addresses, where $n$ is the total number of users.

We also describe a privacy-preserving bridge distribution mechanism for the scenarios where a certain fraction of the distributors may be controlled maliciously by the adversary. 
We stress that \bricks/ can run independently from Tor so that the Tor network can focus on its main purpose of providing anonymity.

The rest of this paper is organized as follows. In Section~\ref{sec:model}, we describe our network and threat model. In Section~\ref{sec:results}, we state our main result as a theorem. We review related work in Section~\ref{sec:relatedwork}. In Section~\ref{sec:algorithm}, we describe our algorithms for reliable bridge distribution; we start from a simple algorithm and improve it as we continue. We describe our implementation of \bricks/ and our simulation results in Section~\ref{sec:simulations}. Finally, we summarize and state our open problems in Section~\ref{sec:conclusion}.

\subsection{Network and Threat Model} \label{sec:model}

In this section, we first define a basic model, where a single distributor performs the bridge distribution task, and then define a multiple distributors model, where a group of distributors collectively run our algorithm. Figure~\ref{fig:model} depicts our high-level network model.

\para{Basic Model.}
We assume there are $n$ \emph{users} (or \emph{clients}) who need to obtain bridge addresses to access Tor. Initially, we assume a single trusted server called the \emph{bridge distributor} (or simply the \emph{distributor}), which has access to a reliable supply of bridge addresses.

We assume an adversary (or \emph{censor}) who can view the internal state and control the actions of $t$ of the clients; we call these adversarially-controlled clients \emph{corrupt users}. \fullpaper{The adversary can corrupt clients probably by hacking into their computers, eavesdropping their communication, or introducing colluding nodes to the network.}
The adversary is \emph{adaptive} meaning that it can corrupt users at any point of the algorithm, up to the point of taking over $t$ users.

The corrupt users have the ability to \emph{block} bridges whose IP addresses they receive. A bridge that is blocked cannot be used by any user. When a bridge is blocked, we assume that the distributor is aware of this fact.\footnote{This can be done using a bridge reachability mechanism discussed in Section~\ref{sec:relatedwork}.}
The adversary does not have to block a bridge as soon as it finds its address; he is allowed to strategically (perhaps by colluding with other corrupt users) decide when to block a bridge. 
We refer to the other ${n-t}$ users as \emph{honest users}. Each honest user wants to obtain a list of bridge addresses, at least one of which is not blocked and hence can be used to connect to Tor. 
We further assume that the adversary has no knowledge of the private random bits used by our algorithm.

We make the standard assumption that there exists a rate-limited channel such as email that the users can use to send their requests for bridges to the distributor, but which is not suitable for interactive Internet communication such as web surfing.\footnote{Completely blocking a service such as email would usually impose major economic and political consequences for censors.} The distributor runs our bridge distribution algorithm locally and sends bridge assignments back to the users via the same channel. \fullpaper{We imagine that all communications over this channel are pseudonymous meaning that they do not reveal any information about the actual identities (i.e., IP addresses) of the sender and the recipient to the either sides of the communication.}

\para{Multiple Distributors Model.} 
We also consider a \emph{multiple distributors model}, where a group of \mbox{$m \ll n$} distributors collectively distribute bridge addresses among the users such that none of the distributors can learn any information about the user-bridge assignments. We assume that the distributors are connected to each other pairwise via a synchronous network with reliable and authenticated channels.

In this model, the adversary not only can corrupt an unknown number of the users, $t$, but can also maliciously control and read the internal state of up to $\lfloor m/3 \rfloor$ of the distributors. The corrupt distributors can deviate from our protocols in any arbitrary manner, \eg, by sending invalid messages or remaining silent. 
We assume that the bridges also use the rate-limited channel for communicating with the distributors with the purpose of registering themselves in the system. Figure~\ref{fig:multidist} shows our multiple distributors model.

\subsection{Our Result} \label{sec:results}
\noindent Below is our main theorem, which we prove in Section~\ref{sec:algorithm}.
\begin{theorem}
	\label{thm:main} There exists a bridge distribution protocol that can run among $m$ distributors and guarantee the following properties with probability ${1 - 1/n^\kappa}$, for some constant ${\kappa \geq 1}$, in the presence of a malicious adversary corrupting at most $\lfloor m/3 \rfloor$ of the distributors: 
	\begin{enumerate}[itemsep=0.4em,topsep=0.5em]
		\item All honest users can connect to Tor after ${\lceil \log{\lceil (t+1)/32 \rceil} \rceil + 1}$ rounds of communication with the distributors;
		\item The total number of bridges required is $O(t\log{n})$;
		\item Each user receives $m$ messages in each round;
		\item Each distributor sends/receives $O(m^2 + n)$ messages;
		\item Each message has length $O(\log{n})$ bits.
	\end{enumerate}
\end{theorem}

We also implement a prototype of \bricks/ and conduct simulations to measure the running time and bridge cost of the protocol. We discuss our experimental results in Section~\ref{sec:simulations}.

\begin{figure}[t]
	\centering
	\includegraphics[width=0.7\linewidth]{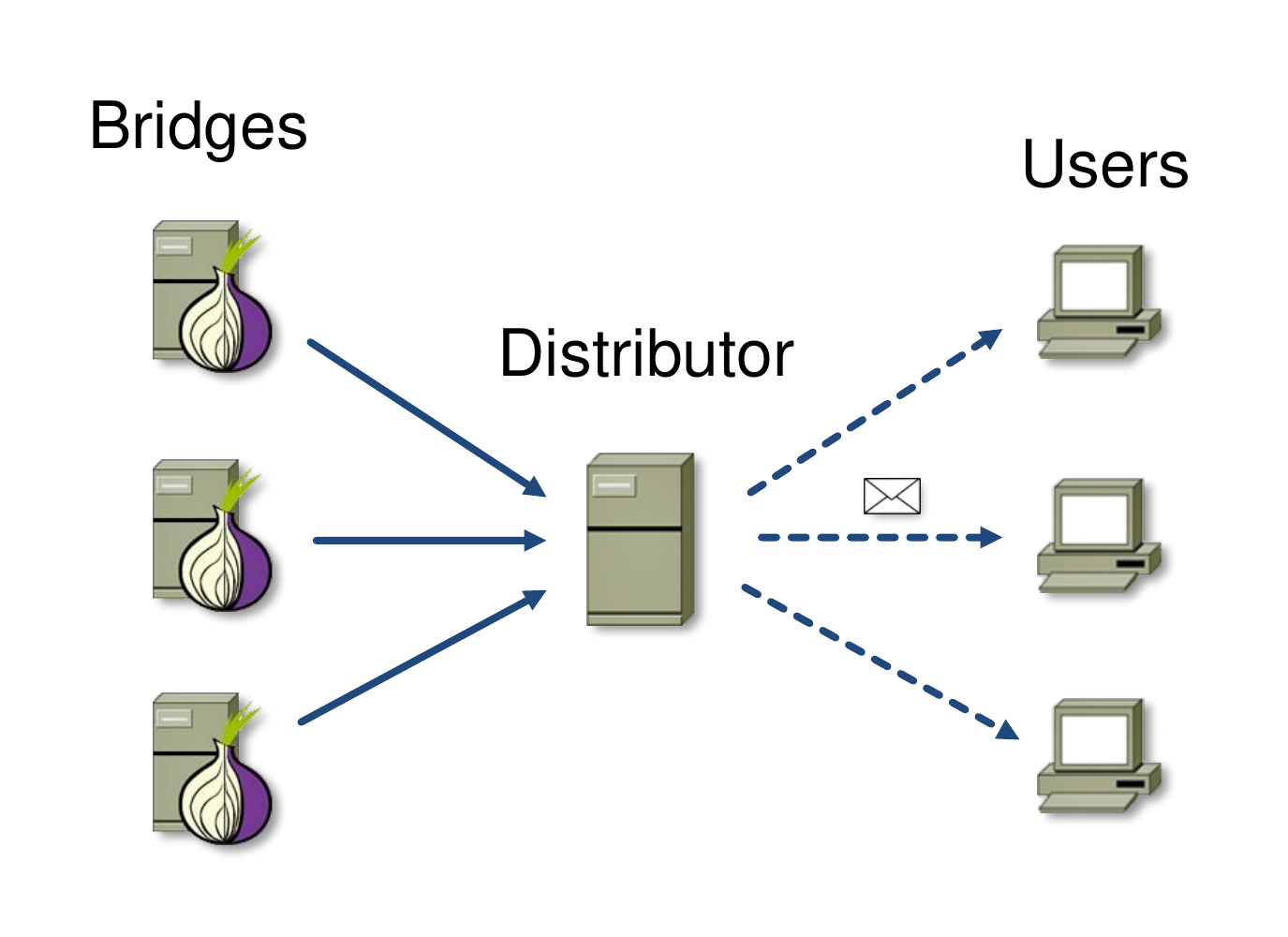}
	\caption{Single distributor model}
	\label{fig:singledist}
\end{figure}


\subsection{Technical Challenges} 
The key technical novelties in the design of \bricks/ are as follows:
\begin{enumerate}[leftmargin=1.7em, itemsep=0.7em, topsep=0.6em]
	\item \sans{Resource-Competitive Costs.} Our protocol adaptively increases the number of bridges distributed among the users with respect to the number of bridges recently blocked by the adversary. This reduces the number of bridges used by the algorithm at the expense of a small (logarithmic) latency cost, which is also a function of the adversary's cost. As a result, the overhead of \bricks/ will always be proportional to the amount of corruption by the adversary. The details are described in Section~\ref{sec:basic-alg}.  
	
	\item \sans{Handling Client Churn.} In practice, users join and leave the system frequently. Adding new users to the system while offering provable robustness against an unknown number of corrupt users is challenging. This is because either (1) the adversary can cause denial of service to the new users if the algorithm's ``next move'' is based on the adversary's behavior; or (2) the protocol cannot guarantee every user receives a usable bridge. We propose a simple technique to handle this with very small (constant) latency overhead. The details are described in Section~\ref{sec:churn}.	
	
	\item \sans{Oblivious Bridge Distribution.} Our technique for computing user-bridge assignments does not depend on actual bridge addresses. Thus, \fullpaper{the distributor does not need to know which bridges are being assigned to which users;} the distributor can \fullpaper{rather} assign ``bridge pseudonyms'' to the users. This prevents an honest-but-curious distributor from snooping on the user-bridge assignments.
	We also show how to distribute these pseudonyms among a group of geographically-dispersed servers who can collectively give the users the information needed to reconstruct their bridge addresses. This protects the anonymity of the users against colluding distributors. The details are described in Sections~\ref{sec:leader-alg} and \ref{sec:decentralized-alg}.
	
	
	\item \sans{Distributed Random Generation.} We show how a \emph{distributed random generation (DRG)} protocol can be used to solve the bridge distribution problem efficiently among multiple untrusted distributors. A DRG protocol allows a group of nodes to collectively generate a uniform random number in a way that none of the nodes can learn it before others, or bias it from the uniform distribution. We employ a well known DRG protocol in \bricks/ to provide the first bridge distribution mechanism that can resist malicious (active) attacks from a subset of the distributors. The details are described in Section~\ref{sec:decentralized-alg}.
\end{enumerate}

\section{Related Work} \label{sec:relatedwork}
\para{Proxy Distribution.} The bridge distribution problem has been studied as the \emph{proxy distribution}, where a set of proxy servers outside a censorship territory are distributed among a set of users inside the territory. These proxies are used to relay Internet traffic to blocked websites. 

\begin{figure}[t]
	\centering
	\includegraphics[width=0.7\linewidth]{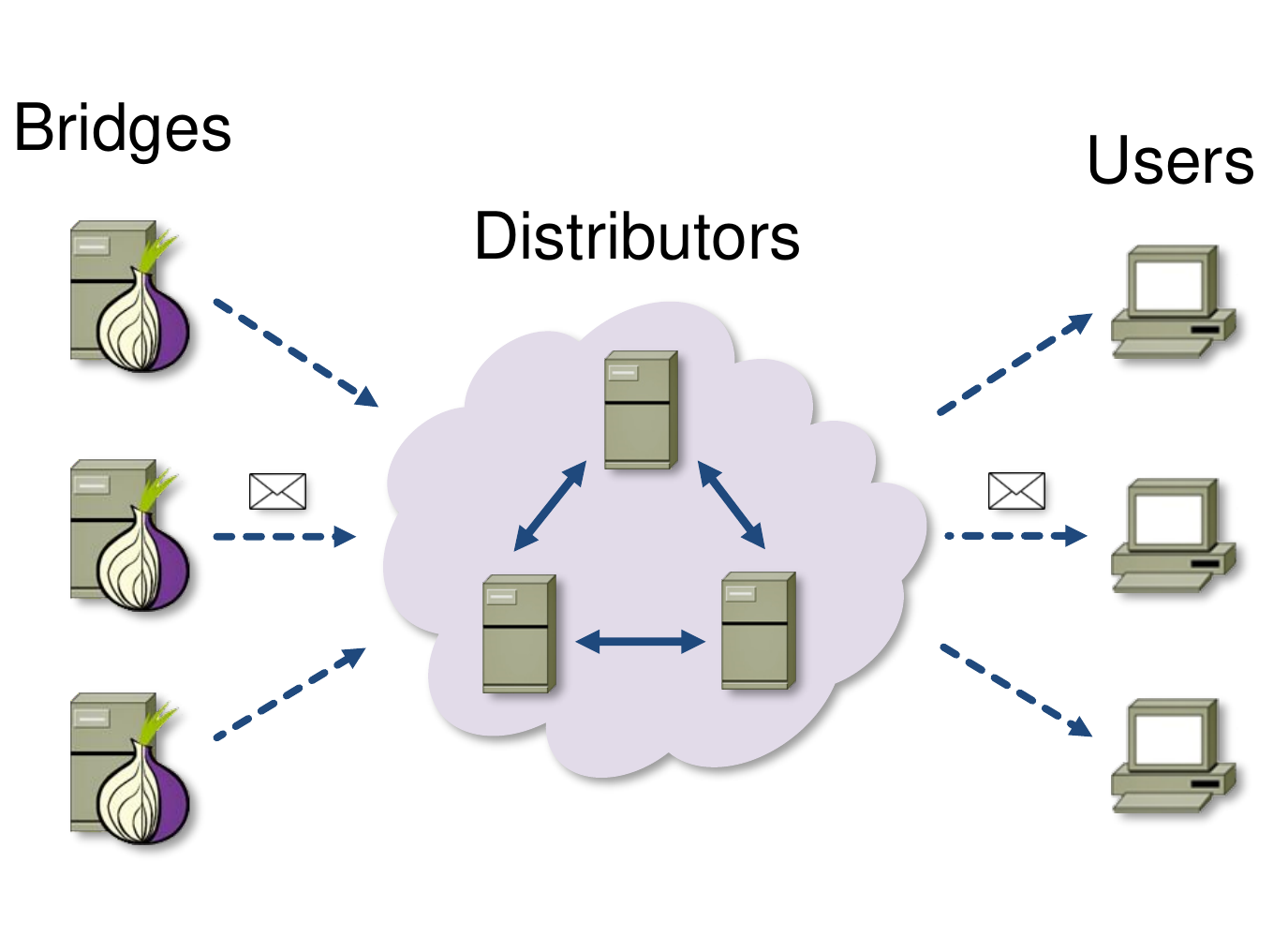}
	\caption{Multiple distributors model}
	\label{fig:multidist}
\end{figure}

Feamster~\etal~\cite{Feamster:PETS:2003} propose a proxy distribution algorithm that requires every user to solve a cryptographic puzzle to discover a proxy. In this way, the algorithm prevents corrupt users from learning a large number of proxies. Unfortunately, empirical results of~\cite{Feamster:PETS:2003} show that a computationally powerful censor can easily block a very large fraction of the proxies.

The Kaleidoscope system of Sovran~\etal~\cite{Sovran:2008:PSN} disseminates proxy addresses over a social network whose links correspond to existing real world social relationships among users. Unfortunately, this algorithm assumes the existence of a few internal trusted users who can relay other users' traffic. In addition, Kaleidoscope cannot guarantee its users' access to Tor.

McCoy~\etal~\cite{McCoy:FC:2011} propose Proximax; a proxy distribution system that uses social networks such as Facebook as trust networks that can provide a degree of protection against discovery by censors. Proximax estimates each user's effectiveness, and chooses the most effective users for advertising proxies, with the goals of maximizing the usage of these proxies while minimizing the risk of having them blocked.

Mahdian~\cite{Mahdian:2010} studies the proxy distribution problem when the number of corrupt users, $t$, is known in advance. He proposes algorithms for both large and small values of $t$ and provides a lower bound for dynamic proxy distribution that is useful only when ${t \ll n}$.
Unfortunately, it is usually hard in practice to reliably estimate the value of $t$. Mahdian's algorithm for large known $t$ requires at most ${t\left(1 + \lceil \log{(n/t)} \rceil \right)}$ bridges, and his algorithm for small known $t$ uses ${O(t^2 \log{n} / \log{\log{n}})}$ bridges.

Wang~\etal~\cite{WangLBH:rBridge:13} propose a reputation-based bridge distribution mechanism called rBridge that computes every user's reputation based on the uptime of its assigned bridges and allows the user to replace a blocked bridge by paying some reputation credits. Interestingly, rBridge is the first model to provide user privacy against an honest-but-curious distributor. This is achieved by performing oblivious transfer between the distributor and the users along with commitments and zero-knowledge proofs for achieving unlinkability of transactions.

\para{Bridge Reachability.} Our algorithm relies on a technique for testing reachability of bridges from outside the censored territory. This assumption is justified via work by Dingledine~\cite{Dingledine:BridgeReach:2011} and Ensafi~\etal~\cite{Ensafi:2014:PAM}, which describe active scanning mechanisms for testing reachability of bridges from outside the censored territory. 

\para{Handling DPI and Active Probing.} The Tor Project has developed a variety of tools known as \emph{pluggable transports}~\cite{Tor:PluggableTransport} to obfuscate the traffic transmitted between clients and bridges. This makes it hard for the censor to perform \emph{deep packet inspection (DPI)} attacks, since distinguishing actual Tor traffic from legitimate-looking obfuscated traffic is hard.

The censor can also block bridges using \emph{active probing}: he can passively monitor the network for suspicious traffic, and then actively probe dubious servers to block those that are determined to run the Tor protocol~\cite{Ensafi2015b}.
We believe that active probing will be defeated in the future using a combination of ideas from CAPTCHAs, port knocking~\cite{PortKnocking2003}, and format transforming encryption~\cite{Dyer:2013:PMM:2508859.2516657}.
Depending on the sophistication of the censor, \bricks/ may be used in parallel with tools that can handle DPI and active probing to provide further protection against blocking.



\para{Resource-Competitive Analysis.} Our analytical approach to bridge distribution can be seen as an application of the \emph{resource-competitive analysis} introduced by Gilbert~\etal~\cite{Gilbert:2012:RAN:2335470.2335471,Bender:2015:SIGACT}. This approach evaluates the performance of any distributed algorithm under attack by an adversary in the following way: if the adversary has a budget of $t$, then the worst-case resource cost of the algorithm is measured by some function of $t$. The adversary's budget is frequently expressed by the number of corrupt nodes controlled by the adversary. This model allows the system to adaptively increase/decrease its resource cost with respect to the \emph{current} amount of corruption by the adversary. Inspired by this model, we design resource-competitive algorithms for bridge distribution that scale reasonably with the adversary's budget.

\section{Our Algorithms} \label{sec:algorithm} 
We first construct a bridge distribution algorithm that is run locally by a single distributor. In Section~\ref{sec:multi-dist}, we extend this algorithm to the multiple distributor model.
We prove the desired properties of these algorithms in Section~\ref{sec:ProofBasic} and Section~\ref{sec:multi-dist} respectively. 
Before proceeding to our algorithms, we define standard terms and notation used in the rest of the paper. 

\para{Notation.} We say an event occurs \emph{with high probability}, if it occurs with probability at least \emph{${1-1/n^\kappa}$}, for some constant ${\kappa \geq 1}$. We denote the set of integers ${\{1,...,n\}}$ by $[n]$, the natural logarithm of any real number $x$ by $\ln{x}$, and the logarithm to the base 2 of $x$ by $\log{x}$. We denote a set of $n$ users participating in our algorithms by ${\{u_1,...,u_n\}}$. We define the \emph{latency} of our algorithm as the maximum number of rounds of communication that any user has to perform with the distributor(s) until he obtains at least one unblocked bridge. \fullpaper{We say a bridge is \emph{blocked} when the censor has restricted users' access to this bridge. We refer to the remaining bridges as \emph{unblocked} bridges.}

\subsection{Basic Algorithm} \label{sec:basic-alg}
Our basic algorithm (shown in Algorithm~\ref{alg:Basic}) is run locally by one distributor.\footnote{By ``run locally'', we mean the distributor computes user-bridge assignments independent of any other distributor and without exchanging any information with them.} 
The algorithm proceeds in \emph{rounds}, where each round corresponds to an increment of the variable $i$ in the while loop.
In each round, the algorithm assigns a set of bridges randomly to a fixed group of users and proceeds to the next round only if the number of blocked bridges exceeds a threshold that is increased in each round. 

The number of bridges distributed in every round is determined based on the threshold in that round as depicted in Figure~\ref{fig:rounds}. If the number of bridges to be distributed in the current round becomes larger than the number of users, $n$, then the algorithm simply assigns a unique unblocked bridge to every user. This happens only if the adversary blocks a large number of bridges and we believe it does not happen in most practical cases. If it happens though, it becomes more reasonable for the algorithm to give each user a unique user.

The exponential growth of the number of bridges distributed in each round allows us to achieve a logarithmic number of rounds (with respect to $t$) until all users can connect to Tor with high probability. In Lemma~\ref{lem:NumIterationsBasic}, we calculate the exact number of rounds required to achieve this goal.

\begin{algorithm}[t]
	\caption{\bricks/ -- Basic Algorithm}
	\label{alg:Basic}
	\vspace{0.4em}
	\textbf{Goal:} Distributes a set of $O(t\log{n})$ bridges among a set of users $\{u_1,...,u_n\}$.
	
	\algFont \vspace{2pt}
	\begin{algorithmic}[1]
		\Statex \hspace{-\algorithmicindent} Run $3\log{n}$ instances of this algorithm in parallel:
		\State Initialize parameters: ${i \gets 0}$; ${b_i \gets 16}$  \label{ln:algstart}
		\While{\True}
			\If {$b_i \geq 0.6 \times 2^{i+4}$} \label{ln:ConditionSimple}
				\State $i \gets i+1$ \label{ln:IncrementSimple}
				\State $d_i \gets 2^{i+4}$
				\If {$d_i < \frac{n}{3\log{n}}$}
					\State $\{B_1,...,B_{d_i}\} \gets$ $d_i$ unblocked bridges \label{ln:RecruitBridges}						
					\ForAll{$j \in [n]$}
						\State Pick $k \in [d_i]$ uniformly at random 
						\State Assign bridge $B_{k}$ to user $u_j$				
					\EndFor			
				\Else
					\State Assign a unique bridge to every user
					\State \textbf{break}
				\EndIf
			\EndIf
			\State 
			\State Check for $b_i \gets$ number of blocked bridges in $\{B_1,...,B_{d_i}\}$
		\EndWhile	\label{ln:algend}
	\end{algorithmic}	
\end{algorithm}

We later show that if one instance of Steps~\ref{ln:algstart}--\ref{ln:algend} of Algorithm~\ref{alg:Basic} is run, then it guarantees that all users can connect to Tor with some constant probability. Therefore, if we repeat these steps ${3\log{n}}$ times, then we can guarantee that all users can connect to Tor with high probability.
If the $3\log{n}$ instances run completely independently, then the adversary can take advantage of this to increase the latency of the algorithm by a factor of $3\log{n}$ using a \emph{serialization attack}: it can strategically coordinate with its corrupt users to block the assigned bridges in such a way that the instances proceed to the next round one at a time. 
We prevent this attack by maintaining a single round counter, $i$, for all instances: whenever the number of blocked bridges in \emph{any} of the instances exceeds the threshold for the current round (\ie,~${0.6 \times 2^{i+4}}$), all instances proceed to the next round.

In every round, \bricks/ only distributes unblocked bridges. To minimize the total number of bridges required, we use unblocked bridges from previous rounds in the current distribution round. This can be done by removing blocked bridges from the pile of previously recruited bridges and adding a sufficient number of new bridges to accommodate the new load.

In each round, \bricks/ sends to every user a single message containing $3\log{n}$ bridge addresses assigned to this user in all  instances of Algorithm~\ref{alg:Basic} that run in parallel. This message is sent to the user via the rate-limited channel (e.g., email). \fullpaper{This is done without requiring the distributor to know the actual IP address of the user. Once the user obtains the set of bridge addresses, he can try to connect to these bridges in parallel to reduce latency.}

\begin{figure*}
	\centering
	\includegraphics[width=0.75\linewidth]{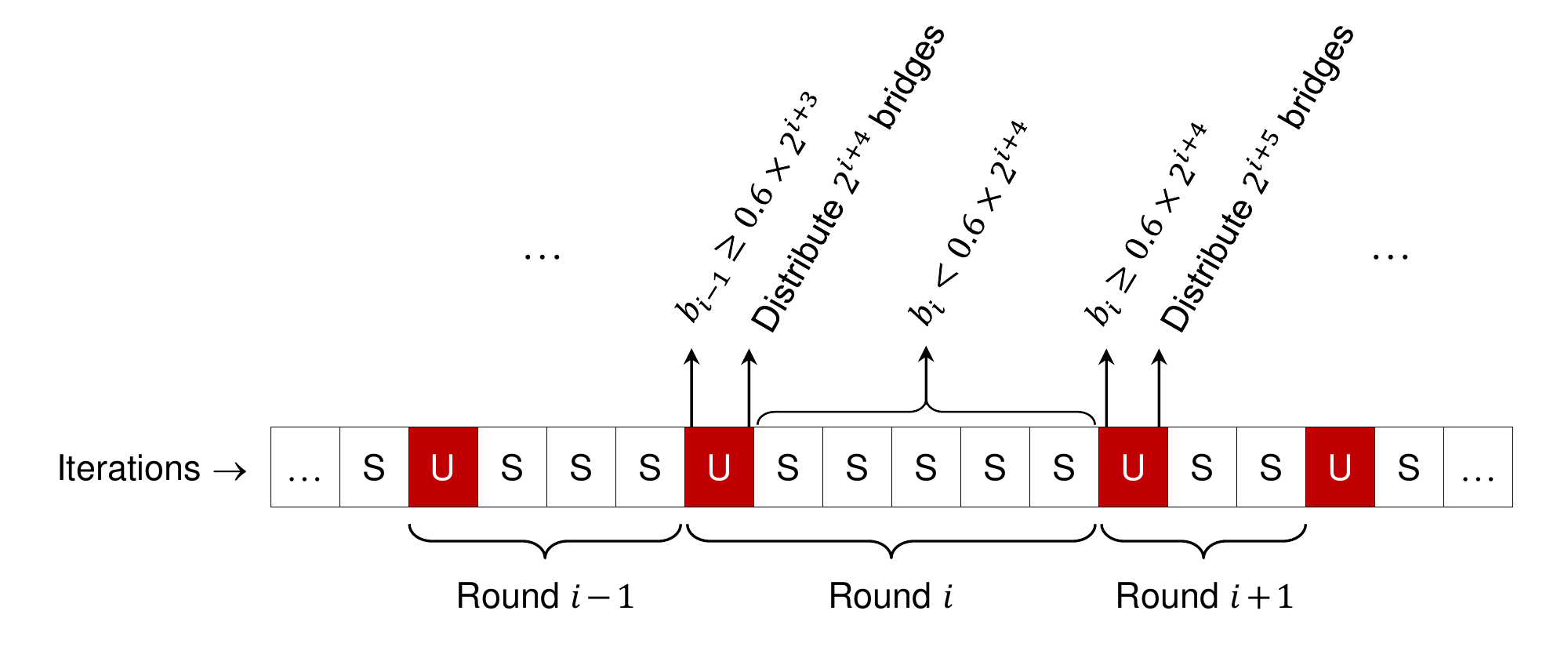}
	\caption{Number of bridges distributed in round $i$ of Algorithm~\ref{alg:Basic}. S and U indicate successful and unsuccessful rounds.}
	\label{fig:rounds}
\end{figure*}

\subsubsection{Analysis of Algorithm~\ref{alg:Basic}} \label{sec:ProofBasic}
We now prove that Algorithm~\ref{alg:Basic} achieves the properties described in Theorem~\ref{thm:main} in the single distributor model. 
We assume a user can connect to Tor in an iteration of the while loop if and only if at least one unblocked bridge is assigned to it. 
Although the adversary has a total budget of $t$ corrupt users, only some of the corrupt users might be actively blocking bridges in any given round. 
Before stating our first lemma, we define the following variables:
\begin{itemize}
	\item $b_i$: number of bridges blocked in round $i$.	
	\item $d_i$: number of bridges distributed in round $i$.
	\item $t_i$: number of corrupt users each of whom has blocked at least one bridge in round $i$.
\end{itemize}

\begin{lemma}[\sans{Robustness}] \label{lem:robustness}
	In round $i$ of Algorithm~\ref{alg:Basic}, if ${b_i < 0.6 \times 2^{i+4}}$, then all honest users can connect to Tor with high probability.
\end{lemma}
\begin{proof}
	We first consider the execution of one of the $3\log{n}$ repeats of Algorithm~\ref{alg:Basic}. For each user, the algorithm chooses a bridge independently and uniformly at random and assigns it to the user. Without loss of generality, assume the corrupt users are assigned bridges first.
	
	For ${k=1,2,...,t_i}$, let $\left\{X_k\right\}$ be a sequence of random variables each representing the bridge assigned to the $k$-th corrupt user. Also, let $Y$ be a random variable corresponding to the number of \emph{bad} bridges (\ie, the bridges that are assigned to at least one corrupt user) after all $t_i$ corrupt users are assigned bridges. The sequence ${\left\{Z_k = \E[Y|X_1,...,X_k]\right\}}$ defines a Doob martingale~\cite[Chapter~5]{dubhashi:2009}, where ${Z_0 = \E[Y]}$. 
	Since each corrupt user is assigned a fixed bridge with probability $1/d_i$, the probability that the bridge is assigned to at least one corrupt user is ${1-(1-1/d_i)^{t_i}}$. By symmetry, this probability is the same for all bridges. Thus, by linearity of expectation,
	\[\E[Y] = \left(1 - \left(1-1/d_i\right)^{t_i}\right)d_i < (1 - e^{-(t_i+1)/d_i})d_i.\]
	
	We know ${t_i < 2^{i+4}}$, because in each round ${d_i = 2^{i+4}}$ bridges are distributed and each corrupt user is assigned exactly one bridge; thus, each corrupt user can block at most one bridge. Hence, 
	\begin{align}
		\E[Y] < (1 - 1/e^{1+1/2^{i+4}})d_i \leq (1 - 1/e^2)d_i \label{eq:expectedBounds}
	\end{align}
	Therefore, in expectation at most a constant fraction of the bridges become bad in each instance of the algorithm. 
	
	Since ${|Z_{k+1} - Z_k| \leq 1}$, ${Z_0 = \E[Y]}$, and ${Z_{t_i} = Y}$, by the Azuma-Hoeffding inequality~\cite[Theorem 5.2]{dubhashi:2009},
	\[\Pr\left(Y > \E[Y] + \lambda\right) \leq e^{\nicefrac{-2\lambda^2}{t_i}},\]
	for any ${\lambda > 0}$. 
	By setting ${\lambda = \sqrt{d_i}}$, we have
	\begin{align}
		\Pr(Y > \E[Y] + \sqrt{d_i}) \leq e^{-2d_i/t_i} < 1/e^2. \label{eq:p1}
	\end{align}
	The last step holds since ${t_i < d_i}$. Therefore, with at most a constant probability, the actual number of bad bridges is larger than its expected value by at most $\sqrt{d_i}$. Therefore, the probability that an honest user is assigned a bad bridge is at most
	\begin{align}
		\frac{\E[Y] + \sqrt{d_i}}{d_i} &< \frac{(1-1/e^{1+1/2^{i+4}})d_i + \sqrt{d_i}}{d_i} \nonumber \\ &= 1/e^{1+1/2^{i+4}} + 1/\sqrt{d_i}, \label{eq:p2}
	\end{align}
	where the first step is achieved using~\eqref{eq:expectedBounds}.
	
	Now, let ${p_1 = \Pr(Y > \E[Y] + \sqrt{d_i})}$, and let $p_2$ be the probability that a fixed honest user is assigned a bad bridge in a fixed instance and a fixed round. From~\eqref{eq:p1} and~\eqref{eq:p2}, we have
	\[p_1 < 1/e^2 \quad \text{and} \quad p_2 < 1/e^{1+1/2^{i+4}} + 1/\sqrt{d_i}.\]
	Thus, the probability that a fixed user fails to receive a good bridge in a fixed instance and a fixed round is equal to ${p_1 + (1-p_1)p_2}$, which is at most $0.6$.
	
	If the algorithm is repeated ${\lceil 3\log{n} \rceil}$ times in parallel, then the probability that the user is assigned to only bad bridges in the last round is at most ${0.6^{\lceil 3\log{n} \rceil} \leq 1/n^2}$.
	By a union bound, the probability that any of the $n$ users is assigned a bad bridge in a round is at most $1/n$. 
	Therefore, all honest users can connect to Tor with high probability.
\end{proof}


Algorithm~\ref{alg:Basic} does not necessarily assign the same number of users to each bridge. However, in the following lemma, we show that each bridge is assigned to almost the same number of users as other bridges with high probability providing a reasonable level of load-balancing.

\begin{lemma}[\sans{Bridge Load-Balancing}]
	Let $X$ be a random variable representing the maximum number of users assigned to any bridge, $Y$ be a random variable representing the minimum number of users assigned to any bridge, and $z = \Theta\left(\frac{\ln{n}}{\ln{\ln{n}}}\right)$. Then, we have 
	\[
	\Pr\left(X \geq \mu z\right) \leq 2/n \quad \text{and} \quad 
	\Pr\left(Y \leq \mu z\right) \leq 2/n,
	\]
	where ${\mu = n/d_i}$.
\end{lemma}
\begin{proof}
	Each round of Algorithm~\ref{alg:Basic} can be seen as the classic balls-and-bins process: $n$ balls (users) are thrown independently and uniformly at random into $d_i$ bins (bridges). It is well known that the distribution of the number of users assigned to a bridge is approximately Poisson with ${\mu = n/d_i}$ \cite[Chapter~5]{Michael2005}.
	
	Let $X_j$ be the random variable corresponding to the number of users assigned to the $j$-th bridge, and let $\tilde{X}_j$ be the Poisson random variable approximating $X_j$. We have ${\mu = \E[X_j] = \E[\tilde{X}_j] = n/d_i}$. We use the following Chernoff bounds from \cite[Chapter~5]{Michael2005} for Poisson random variables:
	\begin{align}
		\Pr(\tilde{X}_j \geq x) \leq e^{-\mu}(e\mu/x)^x \text{, when } x > \mu \label{eq:poissonMax}\\
		\Pr(\tilde{X}_j \leq x) \leq e^{-\mu}(e\mu/x)^x \text{, when } x < \mu \label{eq:poissonMin}
	\end{align}
	
	\noindent Let ${x = \mu y}$, where ${y = ez}$. From \eqref{eq:poissonMax}, we have
	\begin{align}
		\Pr(\tilde{X}_j \geq \mu y) &\leq \left(\frac{e^{y-1}}{y^y}\right)^\mu \nonumber \\
		&\leq \frac{e^{y-1}}{y^y} \nonumber \\ 
		&= \frac{1}{e}\left(\frac{1}{z^z}\right)^e < \frac{1}{n^2}. \label{eq:approxBound}
	\end{align}
	The second step is because ${y^y > e^{y-1}}$ (since ${z > 1}$) and ${\mu > 1}$. The last step is because ${z = \Theta\left(\frac{\ln{n}}{\ln{\ln{n}}}\right)}$ is the solution of ${z^z = n}$. To show this, we take log of both sides of ${z^z = n}$ twice, which yields
	\begin{align*}
		\ln{z} + \ln{\ln{z}} = \ln{\ln{n}}.
	\end{align*}
	We have
	\begin{align*}
		\ln{z} \leq \ln{z} + \ln{\ln{z}} = \ln{\ln{n}} < 2\ln{z}.
	\end{align*}
	Since $z\ln{z} = \ln{n}$,
	\begin{align*}
		z/2 < \frac{\ln{n}}{\ln{\ln{n}}} \leq z.
	\end{align*}
	Therefore, $z = \Theta\left(\frac{\ln{n}}{\ln{\ln{n}}}\right)$. 
	
	It is shown in~\cite[Corollary 5.11]{Michael2005} that for any event that is monotone in the number of balls, if the event occurs with probability at most $p$ in the Poisson approximation, then it occurs with probability at most $2p$ in the exact case. Since the maximum and minimum bridge loads are both monotonically increasing in the number of users, from (\ref{eq:approxBound}) we have
	\begin{align*}
		\Pr(X_j \geq \mu y) \leq 2\Pr(\tilde{X}_j \geq \mu y) < 2/n^2.
	\end{align*}
	
	By applying a union bound over all bridges, the probability that the number of users assigned to any bridge will be more than $\mu z$ is at most $2/n$. The bound on the minimum load can be shown using inequality~(\ref{eq:poissonMin}) in a similar way.
\end{proof}

\begin{lemma}[\sans{Latency}] \label{lem:NumIterationsBasic}
	By running Algorithm~\ref{alg:Basic}, all honest users can connect to Tor with high probability after at most ${\lceil \log{\lceil (t+1)/32 \rceil} \rceil + 1}$ iterations of the while loop.
\end{lemma}
\begin{proof}
	Let $k$ denote the smallest number of rounds required until all users can connect to Tor with high probability. Intuitively, $k$ is  bounded, because the number of corrupt nodes, $t$, is bounded. In the following, we find $k$ with respect to $t$. 
	
	Without loss of generality, we only consider one of the $3\log{n}$ parallel instances of Steps~\ref{ln:algstart}--\ref{ln:algend} of Algorithm~\ref{alg:Basic}. The best strategy for the adversary is to maximize $k$, because this prevents the algorithm from succeeding soon. In each round $i$, this can be achieved by minimizing the number of bridges blocked (\ie, $b_i$), while ensuring the algorithm proceeds to the next round. However, the adversary has to block at least ${0.6 \times 2^{i+4}}$ bridges in each round to force the algorithm to proceed to the next round. Let $\ell$ be the smallest integer such that ${2^\ell \geq t}$. In round $\ell$, the adversary has enough corrupt users to take the algorithm to round ${\ell + 1}$. However, in round ${{\ell + 1}}$, the adversary can block at most ${2^\ell < 2^{\ell+1}}$ bridges, which is insufficient for proceeding to round ${\ell + 2}$. Therefore, ${\ell + 1}$ is the last round and ${k = \ell + 1}$. Since ${2^\ell \geq t}$, and the algorithm starts by distributing 32 bridges, 
	\[k = \lceil \log{\lceil (t+1)/32 \rceil} \rceil + 1.\] 
	In other words, if the while loop runs for at least ${\lceil \log{\lceil (t+1)/32 \rceil} \rceil + 1}$ iterations, then with high probability all honest users can connect to Tor.
\end{proof}

\begin{lemma}[\sans{Bridge Cost}] \label{lem:NumBridgesBasic}
	The total number of bridges used by Algorithm~\ref{alg:Basic} is at most ${(10t + 96)\log{n}}$. 
\end{lemma}
\begin{proof}
	Consider one of the $3\log{n}$ instances of Algorithm~\ref{alg:Basic}. The algorithm starts by distributing  $32$ bridges. In every round ${i > 0}$, the algorithm distributes a new bridge only to replace a bridge blocked in round ${i-1}$. Thus, the total number of bridges used until round $i$, denoted by $M_i$, is equal to the number of bridges blocked until round $i$ plus the number of new bridges distributed in round $i$, which we denote by $a_i$. Therefore,
	\begin{align}
		M_i = a_i + \sum_{j=0}^{i-1} b_j. \label{eq:NumBridges}
	\end{align}
	In round $i$, the algorithm recruits ${a_i \leq 2^{i+4}}$ new bridges, because some of the bridges required for this round might be reused from previous rounds. Since in round $i$ we have ${b_i < 0.6 \times 2^{i+4}}$,
	\[M_i < 2^{i+4} + 0.6\sum_{j=1}^{i-1} 2^{j+4} = 9.6(2^i - 2) + 2^{i+4}\]
	
	From Lemma~\ref{lem:NumIterationsBasic}, it is sufficient to run the algorithm ${k = \lceil \log{\lceil (t+1)/32 \rceil} \rceil + 1}$ rounds. Therefore,	
	\[M_k < 9.6(2^k - 2) + 2^{k+4} \leq 3.2t + 32.\]
	
	Since the algorithm is repeated $3\log{n}$ times, the total number of bridges used by the algorithm is at most ${(10t + 96)\log{n}}$. 
\end{proof}

\subsubsection{Handling Client Churn} \label{sec:churn}
Algorithm~\ref{alg:Basic} can only distribute bridges among a fixed set of users. A more realistic scenario is when users join or leave the algorithm frequently. One way to handle this is to add the new users to the algorithm from the next round (\ie, increment of $i$). This, however, introduces two technical challenges: 
\begin{enumerate}[itemsep=0.5em, topsep=0.6em]
	\item The number of corrupt users is unknown, and hence the adversary can arbitrarily delay the next round, causing a denial of service attack; and
	
	\item Our proof of robustness (Lemma~\ref{lem:robustness}) does not necessarily hold if $n$ is changed, because we repeat the algorithm $3\log{n}$ times to ensure it succeeds with high probability. 
\end{enumerate}

\noindent To handle these challenges, we	 add the following steps to Algorithm~\ref{alg:Basic}:

\begin{enumerate}[itemsep=0.5em, topsep=0.6em]
	\item Each time a user wants to join the system, assign him to $3\log{n}$ random bridges from the set of bridges recruited in the last round (\ie, the last time $i$ was incremented);
	
	\item If the total number of users, $n$, is doubled since the last round, recruit $3 \times 2^{i+4}$ unblocked bridges and assign $3$ of them randomly to each user.
\end{enumerate}

The first step guarantees that the new users are always assigned bridges once they join the system. The second step ensures that the number of parallel instances always remains $3\log{n}$ even if $n$ is changed. This is because $\log{n}$ is increased by one when $n$ is doubled. Therefore, each existing user must receive $3$ new bridges so that the proof of Lemma~\ref{lem:robustness} holds in the setting with churn. Our previous lemmas hold if users leave the system; thus, we only need to update $n$ once they leave.

Since distributing new bridges among existing users is done only after the number of users is doubled, the latency is increased by at most a $\log{n}$ term, where $n$ is the largest number of users in the system during a complete run of the algorithm.

\subsection{Privacy-Preserving Algorithm} \label{sec:multi-dist}
We now adapt Algorithm~\ref{alg:Basic} to the multiple distributor setting. Our goal is to run a protocol jointly among multiple distributors to keep user-bridge assignments hidden from each distributor and from any coalition of up to a $1/3$ 	fraction of them. We assume that a sufficient number of bridges have already registered their email addresses in the system so that in each round the protocol can ask some of them to provide their IP addresses to the system to be distributed by the protocol.

We first construct a \emph{leader-based protocol}, where an honest-but-curious distributor called the \emph{leader} locally runs Algorithm~\ref{alg:Basic} over anonymous bridge addresses. The leader then sends anonymous user-bridge assignments to other distributors who can collectively ``open'' the assignments for the users. 

Next, we construct a fully decentralized protocol, where a group of $m$ distributors collectively compute the bridge distribution functionality while resisting malicious fault from up to a $\lfloor m/3 \rfloor$ fraction of the distributors. Malicious distributors not only may share information with other malicious entities, but also can deviate from our protocol in any arbitrary manner, \eg, by sending invalid messages or remaining silent.

Both of these protocols rely on a secret sharing scheme for the bridges to share their IP addresses with the group of distributors. Before proceeding to our protocols, we briefly describe the secret sharing scheme used in our protocol.

\para{Secret Sharing.} A \emph{secret sharing} protocol allows a party (called \emph{the dealer}) to share a secret among $m$ parties such that any set of $\tau$ or less parties cannot gain any information about the secret, but any set of at least $\tau+1$ parties can reconstruct it. Shamir~\cite{shamir:how} proposed a secret sharing scheme, where the dealer shares a secret $s$ among $m$ parties by choosing a random polynomial $f(x)$ of degree $\tau$ such that ${f(0)=s}$. For all ${j \in [m]}$, the dealer sends $f(j)$ to the $j$-th party. Since at least ${\tau+1}$ points are required to reconstruct $f(x)$, no coalition of $\tau$ or less parties can reconstruct $s$.
The reconstruction algorithm requires a Reed-Solomon decoding algorithm~\cite{Reed-Solomon1960} to correct up to $1/3$ invalid shares sent by dishonest distributors. In our protocol, we use the error correcting algorithm of Berlekamp and Welch~\cite{Berlekamp:Welch:1986}.

\fullpaper{We now briefly describe the reconstruction algorithm. Let $\mathbb{F}_{p}$ denote a finite field of prime order $p$, and $S=\{(x_{1},y_{1})\:|\:x_{j},y_{j}\in\mathbb{F}_{p}\}_{j=1}^{\eta}$ be a set of $\eta$ points, where $\eta-\varepsilon$ of them are on a polynomial $y=P(x)$ of degree $\tau$, and the rest $\varepsilon<(\eta-\tau+1)/2$ points are erroneous. Given the set of points $S$, the goal is to find the polynomial $P(x)$. The algorithm proceeds as follows. Consider two polynomials $E(x)=e_{0}+e_{1}x+...+e_{\varepsilon}x^{\varepsilon}$ of degree $\varepsilon$, and $Q(x)=q_{0}+q_{1}x+...+q_{k}x^{k}$ of degree $k\leq\varepsilon+\tau-1$ such that $y_{i}E(x_{i})=Q(x_{i})$ for all $i\in[\eta]$. This defines a system of $\eta$ linear equations with $\varepsilon+k=\eta$ variables $e_{0},...,e_{\varepsilon},q_{0},...,q_{k}$ that can be solved efficiently using Gaussian elimination technique to get the coefficients of $E(x)$ and $Q(x)$. Finally, calculate $P(x)=Q(x)/E(x)$.}

\subsubsection{Leader-Based Protocol} \label{sec:leader-alg}

Similar to Algorithm~\ref{alg:Basic}, the leader-based protocol also proceeds in rounds. In each round $i$, the leader requests a group of at most $d_i$ bridges to secret-share their IP addresses among all distributors (including the leader) using Shamir's scheme~\cite{shamir:how}. 

Let $(B_1,...B_{d_i})$ denote the sequence of shares the leader receives once the bridges finish the secret sharing protocol. The leader runs Algorithm~\ref{alg:Basic} locally to assign $B_j$'s to the users randomly, for all ${j \in [d_i]}$. Then, the leader broadcasts the pair $(u_k, I_k)$ to all distributors, where $I_k$ is the set of indices of bridges assigned to user $u_k$, for all $k \in [1,...,n]$.

Each distributor then sends its shares of bridge addresses to the appropriate user with respect to the assignment information received from the leader. Finally, each user is able to reconstruct the bridge addresses assigned to him, because at least a 2/3 fraction of the distributors are honest and have correctly sent their shares to the user.
%

\subsubsection{Fully Decentralized Protocol} \label{sec:decentralized-alg}
In each round, Algorithm~\ref{alg:Basic} picks one of the $d_i$ bridges uniformly at random. For each user $u$, if the group of distributors described in the leader-based protocol can collectively agree on a random number ${k \in [d_i]}$, then each of them can individually run Algorithm~\ref{alg:Basic} to assign $u$ to the bridge corresponding to $k$. 

Assuming each distributor holds a share of every bridge address (similar to the leader-based protocol), he can then send his own share to $u$, allowing the user to privately reconstruct the bridge address even if at most a $1/3$ fraction of the shares are invalid.

\para{Distributed Random Generation.} 
We uses a well known commit-and-reveal technique for distributed random generation~\cite{cryptoeprint:2015:366, Tor:DRG:Proposal:2015}. This protocol has at most four rounds of communication and can run efficiently among a small number of distributors to generate unbiased random numbers even if up to a $m/3$ distributors play maliciously.

Let $D_1,...,D_m$ denote the distributors. The DRG protocol starts by each distributor $D_j$ choosing a uniform random number $r_j$ locally, and then publishing a commitment $c_j$ to it. Once all commitments have been received, the distributors reveal their random numbers and verify the commitments. Finally, each distributor computes $r = \sum_{k=1}^m{r_j}$. In \bricks/, we use the simple commitment mechanism described in~\cite{Tor:DRG:Proposal:2015}.

Although this protocol can generate unbiased random numbers, it is vulnerable to \emph{equivocation attacks}\footnote{In \cite{Tor:DRG:Proposal:2015}, these attack are referred to as \emph{partition attacks}.}: A dishonest distributor can send different random values to different distributors while the corresponding commitments are correct and check out. We prevent this by asking the distributors to participate in a Byzantine agreement protocol to agree on the random value $r$ at the end of the protocol. In \bricks/, we use the small scale Byzantine agreement protocol of Castro and Liskov~\cite{CL1} known as BFT. This protocol is efficient for small numbers of participants (about 10) and can tolerate faults from up to a third fraction of the participants. If any of the distributors equivocates, then the BFT protocol fails, and the DRG protocol restarts. 

Another vulnerability is \emph{denial of service attacks}: A dishonest distributor can bias the random number by choosing whether or not to open his commitment; using this he can repeatedly cause the protocol to abort and restart until the resulting value is what he desires. In both equivocation and denial of service attacks, the cheating distributors can be detected easily using administrative mechanisms, especially since the number of distributors is small in our model. Therefore, we believe these attacks do not offer much gain to the adversary.
We finish this section by analyzing the communication complexity of our protocol.
\begin{lemma}[\sans{Communication Complexity}]
	In each round of \bricks/, each user sends/receives at most $m$ messages and each distributor sends/receives ${O(m^2 + n)}$ messages. Each message has length $O(\log{n})$ bits.
\end{lemma}
\begin{proof}
	In the single distributor model, the distributor sends one message to each client per round. Each message contains a list of $3\log{n}$ bridge addresses, therefore it has length $O(\log{n})$ bits.
	
	In the multiple distributors models (leader-based and fully-decentralized), each distributor sends one message to each client per round. Therefore, each client receives $m$ messages in each round. Since each message contains a list of $3\log{n}$ secret-shared values (each corresponding to a bridge address), each message is of size $O(\log{n})$ bits.
	
	In the multiple distributor models, each distributor receives a secret-shared value from each bridge. Since the total number of bridges used by the protocol is $O(t\log{n})$, each distributor receives $O(t\log{n})$ field elements in all rounds from the bridges.	
	
	In the leader-based model, the leader sends to every distributor one message each containing a list of $3\log{n}$ user-bridge information. Therefore, the leader sends a total of $m$ messages each of size $O(\log{n})$ bits in each round. Each distributor in this model sends to each client one message each containing $3\log{n}$ secret-shared values, thus each distributor sends/receives \[O\left(m + n + \frac{t\log{n}}{\log{t}}\right) = O(m + n)\] messages of size $O(\log{n})$ bits in each round.
	
	In the fully decentralized model, each distributor participates in a run of the random generation protocol per round. This protocol consists of one secret sharing round transmitting $m$ field elements per distributor, and one run of the BFT protocol, which sends $O(m^2)$ field elements per distributor. Finally, each distributor sends to each client one message containing a list of $3\log{n}$ secret-shared values. Thus, each distributor sends/receives \[O\left(m^2 + n + \frac{t\log{n}}{\log{t}}\right) = O(m^2 + n)\] messages each of size $O(\log{n})$ bits in each round.
	
\end{proof}

\section{Evaluation} \label{sec:simulations}
We implemented a proof-of-concept prototype and tested it in a simulated environment under various adversarial behavior. The prototype is written in C\# using .NET Framework 4.5. We ran the simulations on an Intel Core i5-4250U 1.3GHz machine with 4GB of RAM running Windows 10 Pro. We set the parameters of \bricks/ in such a way that we ensure it fails with probability at most $10^{-4}$. 

\begin{figure*}
	\hspace{-0.8em}\includegraphics[width=0.34\textwidth]{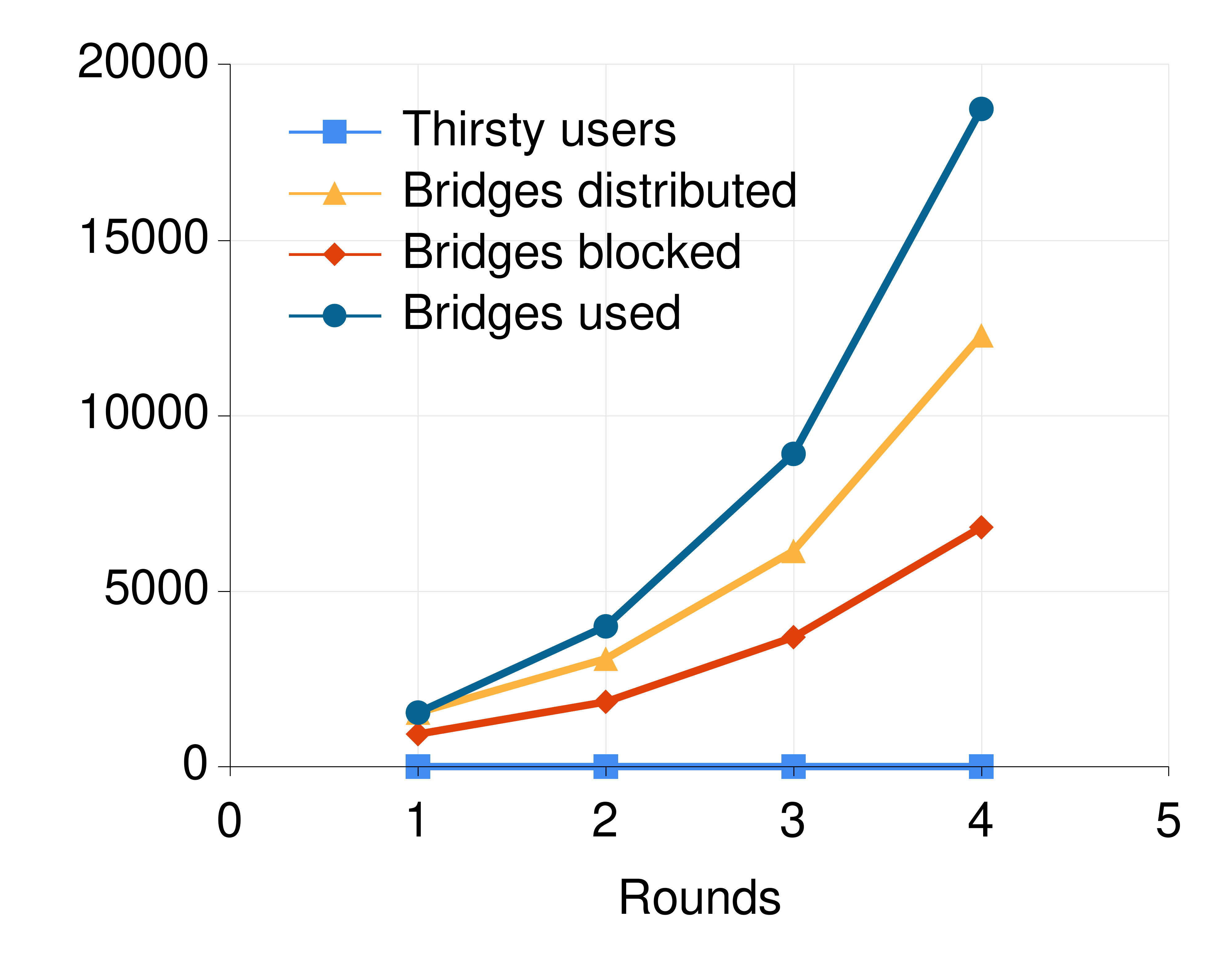}
	\hspace{-0.5em}\includegraphics[width=0.34\textwidth]{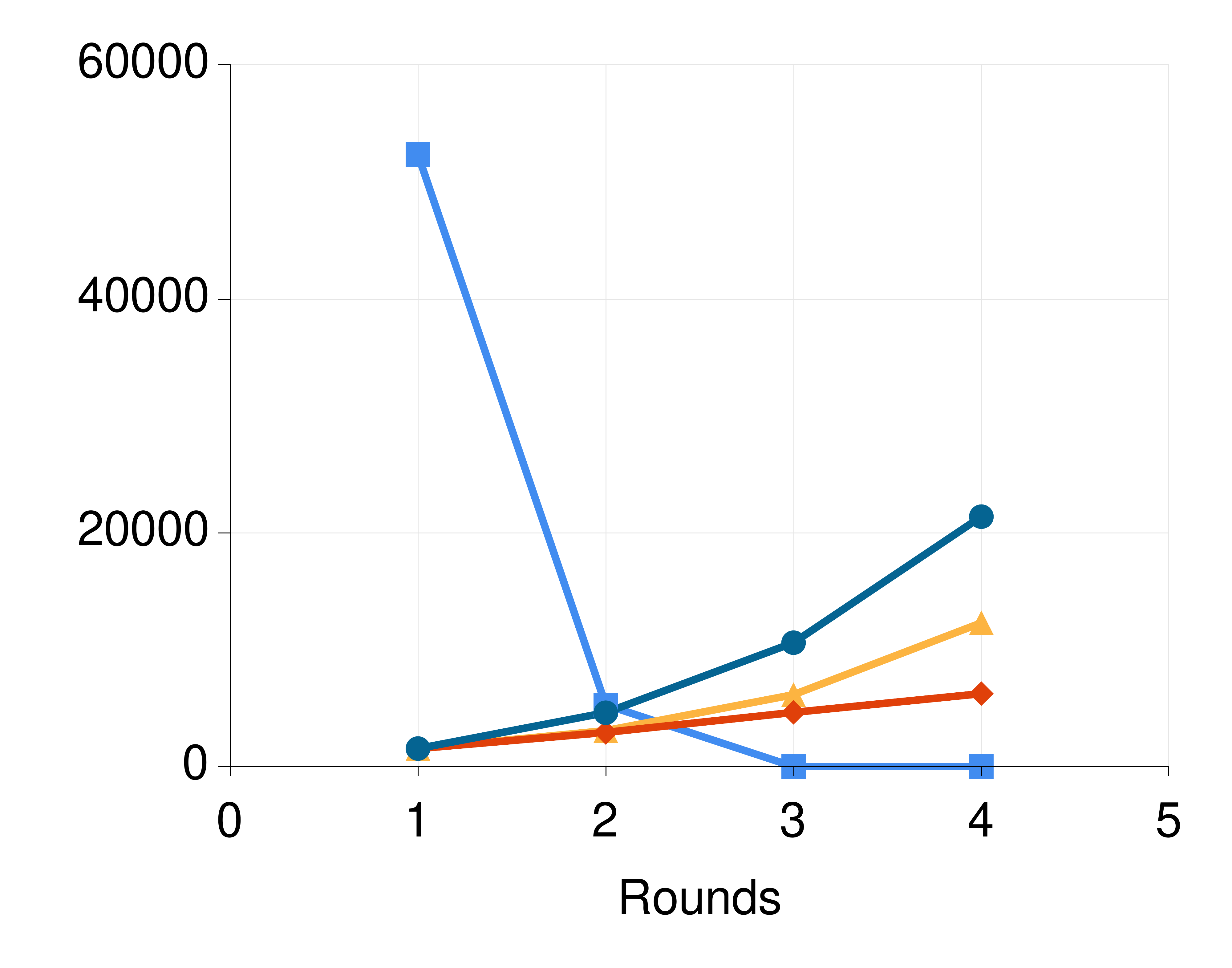}
	\hspace{-0.5em}\includegraphics[width=0.34\textwidth]{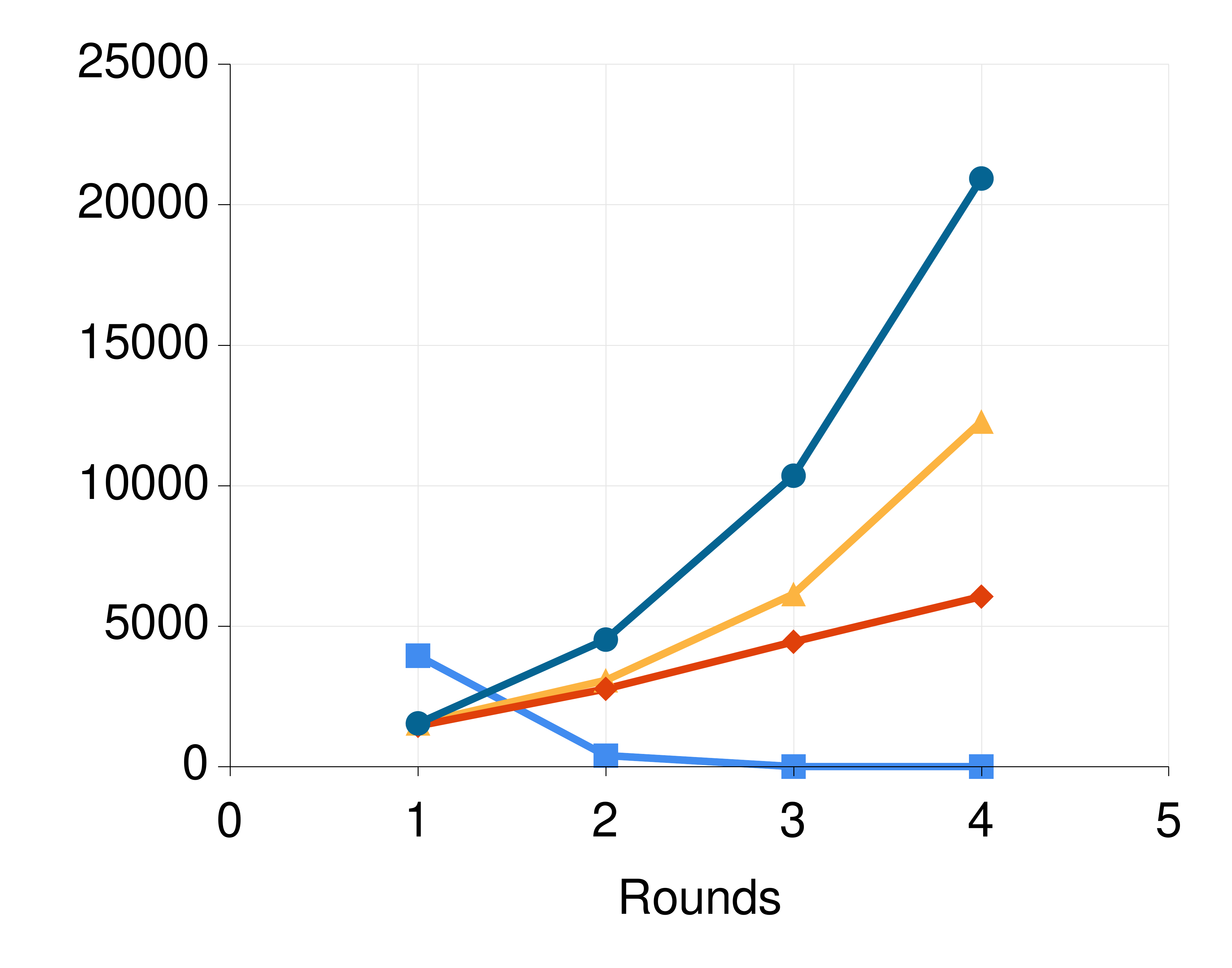}
	\caption{Simulation results for ${n=65536}$ and ${t=180}$ with prudent (left), aggressive (middle), and stochastic (right) adversary.}
	\label{fig:plot1} 
\end{figure*}

\begin{figure*}[t]
	\hspace{-0.4em}
	\includegraphics[width=0.32\textwidth]{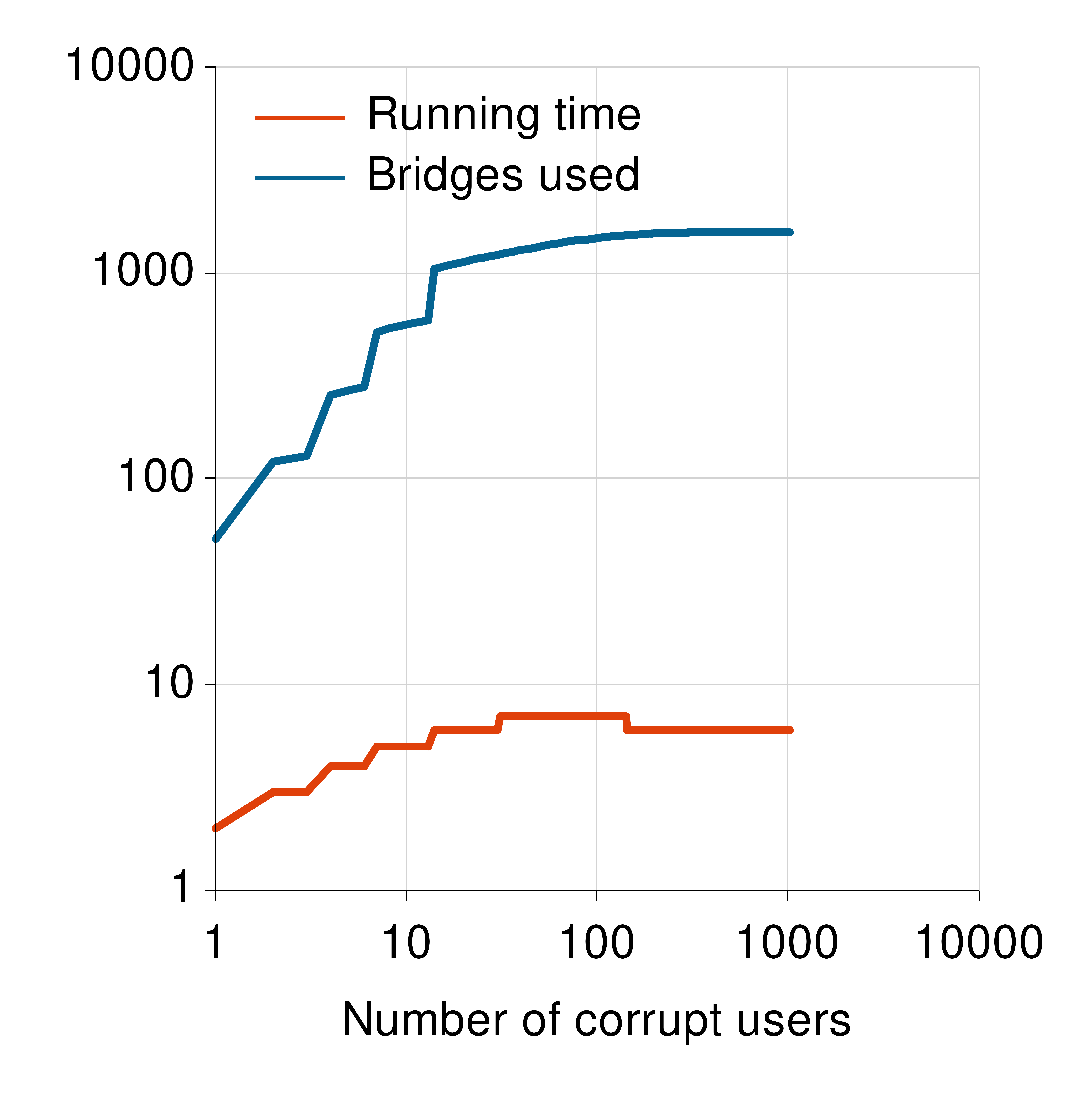}	
	\hspace{0.3em}
	\includegraphics[width=0.32\textwidth]{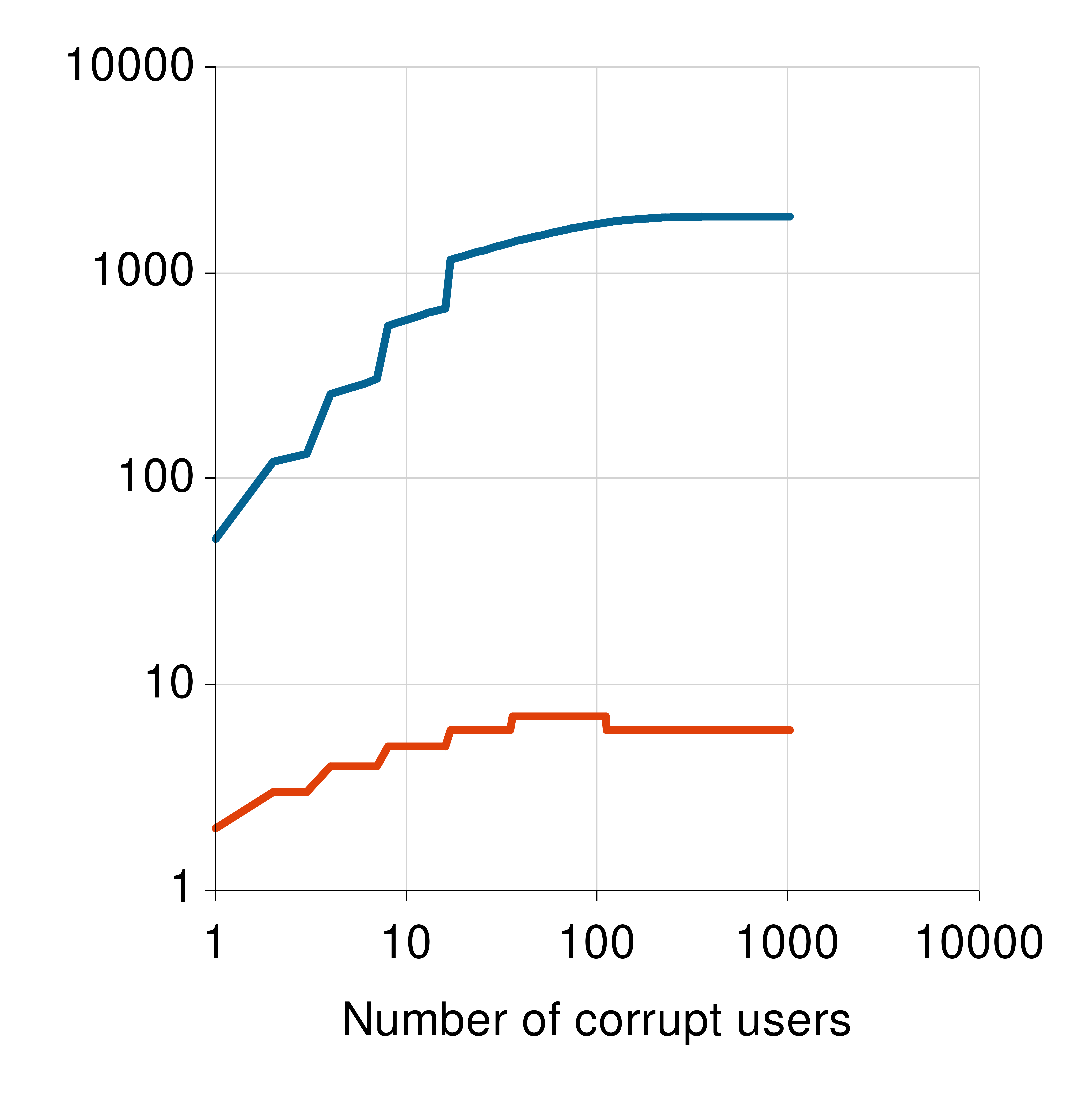}	
	\hspace{0.3em}
	\includegraphics[width=0.32\textwidth]{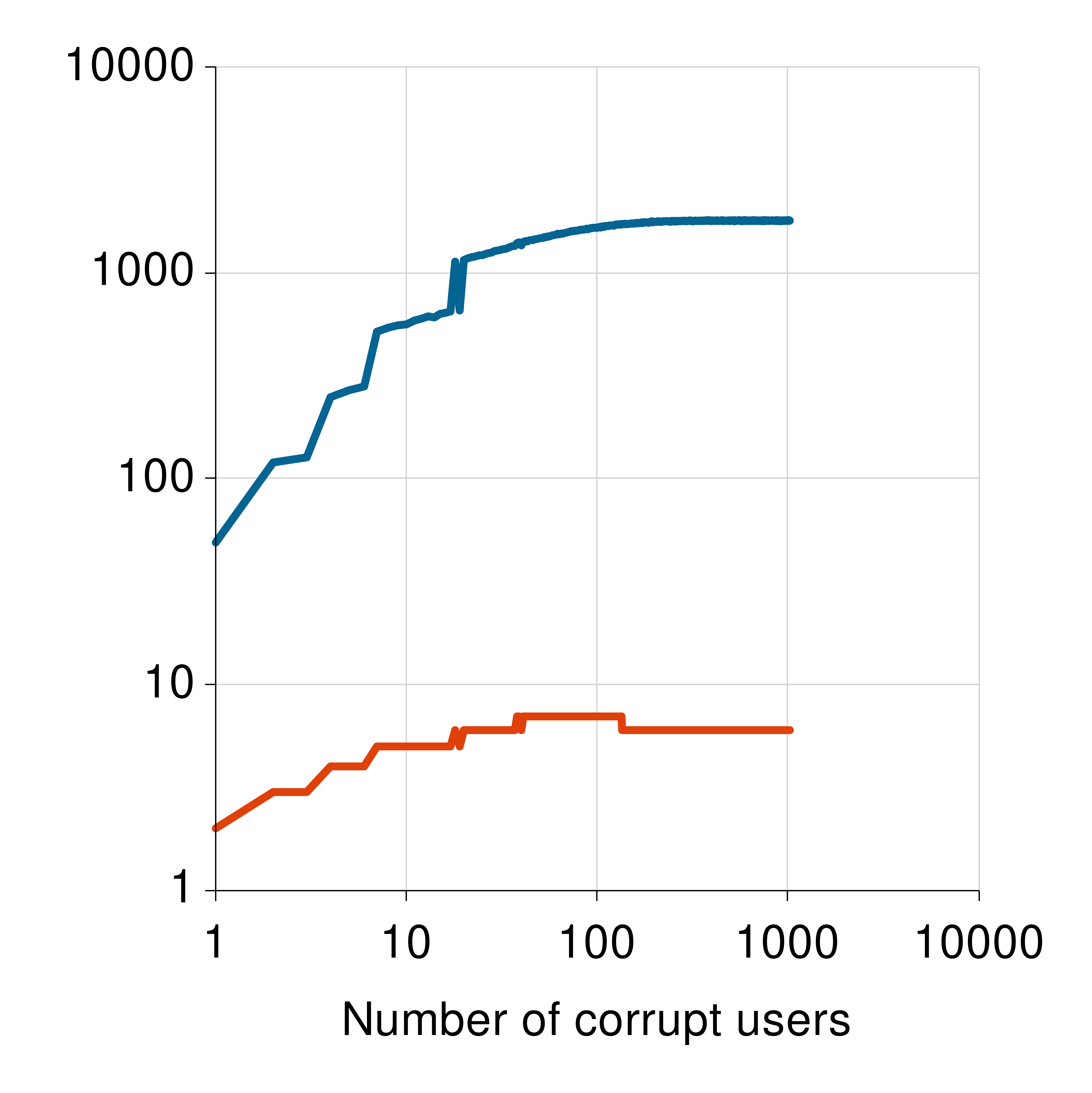}
	\caption{Simulation results for ${n=1024}$ and variable number of corrupt users with prudent (left), aggressive (middle), and stochastic (right) adversary.}
	\label{fig:plot2} 
\end{figure*}


We consider three blocking strategies for the adversary: \emph{prudent}, \emph{aggressive}, and \emph{stochastic}. A prudent adversary blocks the minimum number of bridges in each round such that the algorithm is forced to go to the next round. An aggressive adversary blocks immediately all of the bridges he learns from the corrupt users. Finally, a stochastic adversary blocks each bridge he receives with some fixed probability.

To evaluate the performance of \bricks/, we calculate five measures of performance in two experiments. These measures are:

\begin{enumerate}[itemsep=0.4em, topsep=0.55em]
	\item \sans{Thirsty Users:} Number of users who do not have any unblocked bridge in the current round. \label{measure:thirsty}
	\item \sans{Bridges Distributed:} Number of bridges distributed in the current round ($d_i$). \label{measure:d_i}
	\item \sans{Bridges Blocked:} Number of bridges blocked in the current round ($b_i$). \label{measure:b_i}
	\item \sans{Bridges Used:} Total number of unique bridges distributed by the algorithm until this round ($M_i$). \label{measure:M_i}
	\item \sans{Latency:} Number of round until all users receive at least one unblocked bridge. \label{measure:latency}
\end{enumerate}

In the first experiment (shown in Figure~\ref{fig:plot1}), we run the algorithm in the basic algorithm for ${n=65536}$ and ${t=180}$, and calculate Measures~\ref{measure:thirsty}--\ref{measure:M_i} at the end of each round after running the algorithm over 10 samples for a fixed set of parameters. The experiment was run with the three different adversarial strategies. For the stochastic blocking, the adversary blocks each bridge with probability $0.95$. 
In the second experiment (shown in Figure~\ref{fig:plot2}), we run the algorithm using a single distributor for ${n=1024}$ and calculate Measures~\ref{measure:M_i} and \ref{measure:latency} by varying $t$ between $0$ and $1023$.

Our results indicate that \bricks/ incurs a very small cost when there is small or no corruption in the network. Moreover, the algorithm scales well with the number of corruptions and can quickly adapt to the adversary's behavior. These all support our claim that \bricks/ can be used for practical bridge distribution with provable access guarantee for all users.

\section{Conclusion} \label{sec:conclusion}
We described \bricks/, a bridge distribution system that allows all honest users to connect to Tor in the presence of an adversary corrupting an unknown number of users. Our algorithm can adaptively increase the number of bridges according to the behavior of the adversary and use near-optimal number bridges. We also modified our algorithm slightly to handle user churn (join/leave) by adding small (constant) amortized latency.

We also described a protocol for privacy-preserving bridge distribution by running the distribution algorithm obliviously among a group of distributors. We showed that the resulting protocol not only can protect the privacy of user-bridge assignments from any coalition of up to a $1/3$ fraction of the distributors, but also can tolerate malicious attacks from a $1/3$ fraction of the distributors. We finally evaluated a prototype of our protocol in different simulated scenarios to show that \bricks/ can be used efficiently in practice.

Although \bricks/ represents a step towards robust and privacy-preserving bridge distribution, many challenges remain for future work. For example, the current algorithm uses a relatively large number of bridges when the number of corrupt users is large. Is it possible to make the bridge cost sublinear in $t$ with practical constant terms? 
An interesting direction for answering this question to use inexpensive honeypot bridges for detecting and blacklisting corrupt users. This, however, requires a mechanism such as CAPTCHA for preventing the adversary from distinguishing real bridges from the fake ones. Moreover, a colluding adversary may be able to compare bridges assigned to its corrupt users to detect honeypots.

To better explore the possibility of achieving a sublinear bridge cost, one may consider finding lower bounds for different scenarios. For example, when each user is assigned at least one bridge, it seems impossible to achieve a sublinear bridge cost unless some of the bridges are fake, or we only distribute real bridges in random-chosen rounds. What is the lower bound for the number of rounds in these scenarios?

Another interesting open problem is to examine if our current notion of robustness is overkill for practice. For example, is it possible to significantly reduce our costs by guaranteeing access for all but a constant number of users? 

\bibliographystyle{plain}
\bibliography{security}

\begin{thebibliography}{10}

\bibitem{Oni:2012:China}
{The Open Net Initiative}: China.
\newblock URL: \url{https://goo.gl/YLb37p}, 2012.

\bibitem{Tor:Users}
The {Tor Project} metrics: Direct users connecting between {January} 1, 2015
  and {March} 31, 2015.
\newblock URL: \url{https://goo.gl/mz1vLS}, 2015.

\bibitem{Tor:Relays}
The {Tor Project} metrics: Relays in the network between {January} 1, 2015 and
  {March} 31, 2015.
\newblock URL: \url{https://goo.gl/Cs7R0Y}, 2015.

\bibitem{Tor:PluggableTransport}
The {Tor Project}: Pluggable transport.
\newblock URL: \url{https://goo.gl/SBGupD}, 2015.

\bibitem{Tor:Bridges}
The {Tor Project} metrics: Bridges in the network between {March} 1, 2016 and
  {March} 31, 2016.
\newblock URL: \url{https://goo.gl/Cs7R0Y}, 2016.

\bibitem{Bender:2015:SIGACT}
Michael~A. Bender, Jeremy~T. Fineman, Mahnush Movahedi, Jared Saia, Varsha
  Dani, Seth Gilbert, Seth Pettie, and Maxwell Young.
\newblock Resource-competitive algorithms.
\newblock {\em ACM SIGACT News}, 46(3):57--71, September 2015.

\bibitem{Berlekamp:Welch:1986}
E~Berlekamp and L~Welch.
\newblock Error correction for algebraic block codes, {US Patent} 4,633,470,
  December 1986.

\bibitem{CL1}
Miguel Castro and Barbara Liskov.
\newblock {Practical {B}yzantine Fault Tolerance}.
\newblock In {\em Proceedings of the Symposium on Operating Systems Design and
  Implementation (OSDI)}, 1999.

\bibitem{Dingledine:BridgeReach:2011}
Roger Dingledine.
\newblock Research problem: Five ways to test bridge reachability.
\newblock URL: \url{https://goo.gl/BTJuZP}, 2011.

\bibitem{Dingledine:Bridges:2011}
Roger Dingledine.
\newblock Research problems: Ten ways to discover {Tor} bridges.
\newblock URL: \url{https://goo.gl/CYFfC4}, 2011.

\bibitem{Dingledine06designof}
Roger Dingledine and Nick Mathewson.
\newblock Design of a blocking-resistant anonymity system.
\newblock Technical report, The Tor Project Inc., 2006.

\bibitem{dingledine:2004}
Roger Dingledine, Nick Mathewson, and Paul Syverson.
\newblock {Tor}: the second-generation onion router.
\newblock In {\em Proceedings of the 13th USENIX Security Symposium}, Berkeley,
  CA, USA, 2004.

\bibitem{dubhashi:2009}
Devdatt~P. Dubhashi and Alessandro Panconesi.
\newblock {\em Concentration of Measure for the Analysis of Randomized
  Algorithms}.
\newblock Cambridge University Press, New York, NY, USA, 2009.

\bibitem{Dyer:2013:PMM:2508859.2516657}
Kevin~P. Dyer, Scott~E. Coull, Thomas Ristenpart, and Thomas Shrimpton.
\newblock Protocol misidentification made easy with format-transforming
  encryption.
\newblock In {\em Proceedings of the 2013 ACM SIGSAC Conference on Computer
  \&\#38; Communications Security}, CCS '13, pages 61--72, New York, NY, USA,
  2013. ACM.

\bibitem{Ensafi2015b}
Roya Ensafi, David Fifield, Philipp Winter, Nick Feamster, Nicholas Weaver, and
  Vern Paxson.
\newblock Examining how the {G}reat {F}irewall discovers hidden circumvention
  servers.
\newblock In {\em Internet Measurement Conference (IMC)}. ACM, 2015.

\bibitem{Ensafi:2014:PAM}
Roya Ensafi, Jeffrey Knockel, Geoffrey Alexander, and Jedidiah~R. Crandall.
\newblock Detecting intentional packet drops on the {I}nternet via {TCP/IP}
  side channels.
\newblock In {\em Proceedings of the 15th International Conference on Passive
  and Active Measurement - Volume 8362}, PAM 2014, pages 109--118, New York,
  NY, USA, 2014. Springer-Verlag New York, Inc.

\bibitem{Feamster:PETS:2003}
Nick Feamster, Magdalena Balazinska, Winston Wang, Hari Balakrishnan, and David
  Karger.
\newblock Thwarting web censorship with untrusted messenger discovery.
\newblock In Roger Dingledine, editor, {\em Privacy Enhancing Technologies},
  volume 2760 of {\em Lecture Notes in Computer Science}, pages 125--140.
  Springer Berlin Heidelberg, 2003.

\bibitem{Gilbert:2012:RAN:2335470.2335471}
Seth Gilbert, Jared Saia, Valerie King, and Maxwell Young.
\newblock Resource-competitive analysis: A new perspective on attack-resistant
  distributed computing.
\newblock In {\em Proceedings of the 8th International Workshop on Foundations
  of Mobile Computing}, FOMC '12, pages 1:1--1:6, New York, NY, USA, 2012. ACM.

\bibitem{Tor:DRG:Proposal:2015}
David Goulet and George Kadianakis.
\newblock Random number generation during tor voting.
\newblock Tor's Protocol Specifications -- Proposal 250, August 2015.
\newblock URL: \url{https://goo.gl/MfZkzD}.

\bibitem{PortKnocking2003}
Martin Krzywinski.
\newblock Port knocking: Network authentication across closed ports.
\newblock Technical report, 2003.

\bibitem{cryptoeprint:2015:366}
Arjen~K. Lenstra and Benjamin Wesolowski.
\newblock A random zoo: sloth, unicorn, and trx.
\newblock Cryptology ePrint Archive, Report 2015/366, 2015.
\newblock \url{http://eprint.iacr.org/}.

\bibitem{Ling:2012:infocom}
Z.~Ling, J.~Luo, W.~Yu, M.~Yang, and X.~Fu.
\newblock Extensive analysis and large-scale empirical evaluation of tor bridge
  discovery.
\newblock In {\em INFOCOM, 2012 Proceedings IEEE}, pages 2381--2389, March
  2012.

\bibitem{Mahdian:2010}
Mohammad Mahdian.
\newblock Fighting censorship with algorithms.
\newblock In Paolo Boldi and Luisa Gargano, editors, {\em Fun with Algorithms},
  volume 6099 of {\em Lecture Notes in Computer Science}, pages 296--306.
  Springer Berlin Heidelberg, 2010.

\bibitem{McCoy:FC:2011}
Damon McCoy, Jose~Andre Morales, and Kirill Levchenko.
\newblock Proximax: Measurement-driven proxy dissemination.
\newblock In {\em Proceedings of the 15th International Conference on Financial
  Cryptography and Data Security}, FC'11, pages 260--267, Berlin, Heidelberg,
  2012. Springer-Verlag.

\bibitem{Michael2005}
Michael Mitzenmacher and Eli Upfal.
\newblock {\em {Probability and Computing: Randomized Algorithms and
  Probabilistic Analysis}}.
\newblock Cambridge University Press, 2005.

\bibitem{Reed-Solomon1960}
Irving Reed and Gustave Solomon.
\newblock Polynomial codes over certain finite fields.
\newblock {\em Journal of the Society for Industrial and Applied Mathematics
  (SIAM)}, pages 300--304, 1960.

\bibitem{Rushe:2012:Censorship}
Dominic Rushe.
\newblock Google reports 'alarming' rise in censorship by governments.
\newblock The Guardian, June 2012.

\bibitem{shamir:how}
Adi Shamir.
\newblock How to share a secret.
\newblock {\em Commun. ACM}, 22(11):612--613, 1979.

\bibitem{Sovran:2008:PSN}
Yair Sovran, Alana Libonati, and Jinyang Li.
\newblock Pass it on: Social networks stymie censors.
\newblock In {\em Proceedings of the 7th International Conference on
  Peer-to-peer Systems}, IPTPS'08, pages 3--3, Berkeley, CA, USA, 2008. USENIX
  Association.

\bibitem{Turner:2016:Surveillance}
Karen Turner.
\newblock Mass surveillance silences minority opinions, according to study.
\newblock The Washington Post, March 2016.

\bibitem{WangLBH:rBridge:13}
Qiyan Wang, Zi~Lin, Nikita Borisov, and Nicholas Hopper.
\newblock rbridge: User reputation based tor bridge distribution with privacy
  preservation.
\newblock In {\em Network and Distributed System Security Symposium}, NDSS
  2013. The Internet Society, 2013.

\bibitem{BridgeBlockingChina:2012}
Philipp Winter and Stefan Lindskog.
\newblock How the great firewall of {China} is blocking {Tor}.
\newblock In {\em 2nd USENIX Workshop on Free and Open Communications on the
  Internet}, Berkeley, CA, 2012.

\end{thebibliography}


\end{document}